\documentclass[conference]{IEEEtran}
\IEEEoverridecommandlockouts
\usepackage{cite}
\usepackage{amsmath,amssymb,amsfonts}
\usepackage{algorithmic}
\usepackage{graphicx}
\usepackage{textcomp}
\usepackage{xcolor}
\usepackage{cleveref}
\usepackage{multirow}

\newcounter{remarkcounter}

\newenvironment{remark}
{
	\refstepcounter{remarkcounter} 
	\par \textit{Remark \theremarkcounter:}} 
{\par}

\newcounter{theoremcounter}

\newenvironment{theorem}
{
	\refstepcounter{theoremcounter} 
	\par \textit{Theorem \thetheoremcounter:}} 

{\par}

\newenvironment{proof}
{
	\par \textit{Proof:}} 
{\hfill$\blacksquare$\par}

\def\BibTeX{{\rm B\kern-.05em{\sc i\kern-.025em b}\kern-.08em
    T\kern-.1667em\lower.7ex\hbox{E}\kern-.125emX}}
\begin{document}

\title{Dual-Channel Adaptive NMPC for Quadrotor under Instantaneous Impact and Payload Disturbances\\
\thanks{This work was supported in part by the National Natural Science Foundation of China under Grant 62473292 and Grant 62088101, and in part by the Shanghai Municipal Science and Technology
Major Project under Grant 2021SHZDZX0100.}
}

\author{\IEEEauthorblockN{1\textsuperscript{st} Xinqi Chen}
	\IEEEauthorblockA{\textit{CEIE}\\
		\textit{Tongji University}\\
		Shanghai 201804, China\\
		2331868@tongji.edu.cn}
	\and
	\IEEEauthorblockN{2\textsuperscript{nd} Xiuxian Li}
	\IEEEauthorblockA{\textit{CEIE, SRIAS}\\
		\textit{Tongji University}\\
		Shanghai 201804, China\\
		xli@tongji.edu.cn}
	\and
	\IEEEauthorblockN{3\textsuperscript{rd} Min Meng}
	\IEEEauthorblockA{\textit{CEIE, SRIAS}\\
		\textit{Tongji University}\\
		Shanghai 201804, China\\
		mengmin@tongji.edu.cn}

}
\maketitle
\begin{abstract}
Capturing target objects using the quadrotor has gained increasing popularity in recent years, but most studies focus on capturing lightweight objects. The instantaneous contact force generated when capturing objects of a certain mass, along with the payload uncertainty after attachment, will pose significant challenges to the quadrotor control. This paper proposes a novel control architecture, namely Dual-Channel Adaptive Nonlinear Model Predictive Control (DCA-NMPC), which cascades a nonlinear model predictive control with two lower-level model reference adaptive controllers and can resist drastic impact and adapt to uncertain inertial parameters. Numerical simulation experiments are performed for validation.
\end{abstract}

\section{Introduction}
Quadrotors have been extensively researched due to their promising application prospects \cite{kumar2012opportunities}. In recent years, numerous studies have focused on equipping quadrotors with rackets or nets to hit or capture balls, which are important for tasks involving the capture of falling targets \cite{muller2011quadrocopter,su2017catching,yu2023catch}. However, most of them focus on capturing lightweight balls, which hardly affect the quadrotor's control. When the quadrotor captures the object of a certain mass, like a mobile phone, tea cup, etc., the target object will generate instantaneous contact force on the quadrotor, changing its velocity and angular rate abruptly. Additionally, as the target object attaches to the quadrotor after capture, the mass, inertia matrix, and center of mass (CoM) position of the quadrotor are altered, and these disturbances undoubtedly pose challenges to the precise and timely control of the quadrotor. Therefore, how to react rapidly to disturbances and resist parameter uncertainty becomes crucial.

\begin{figure}[tbp]
	\centering
	\includegraphics[scale=0.7, trim=1cm 0.3cm 0.5cm 0cm, clip]{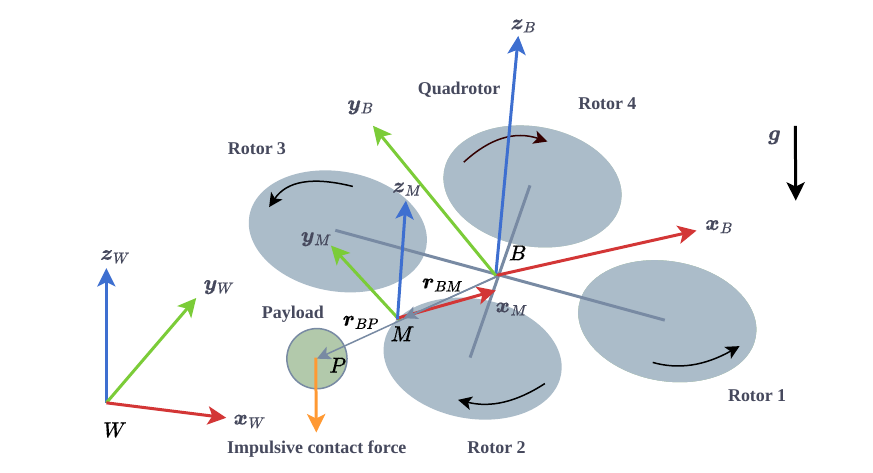}%
	\caption{Diagram of the quadrotor-payload system. Three coordinate frames are defined: world frame $W$, body frame $B$, and center of mass frame $M$. Impulsive contact force and payload uncertainty are considered in this work.}
	\label{fig:frame}
	\vspace{-15pt}
\end{figure}

Integrating uncertainties and external disturbances into a lump disturbance term for processing is common in controller design, which can be combined with the sliding mode controller \cite{hassani2024model} and the disturbance observer \cite{xu2024novel}. Such controllers are applicable to most disturbance rejection scenarios but lack the exploitation of the specific characteristics of the disturbances. \cite{jia2023evolver} combines the disturbance observer with the Koopman operator to learn disturbance dynamics online from data. However, the lifting function, crucial to control performance, lacks reliable design principles. Model reference adaptive control (MRAC) is widely applied to identification and control problems under parametric uncertainty \cite{nguyen2018model,lavretsky2012robust}. In \cite{maki2020model}, a nonlinear multi-input multi-output MRAC adapting to the changes of the inertia matrix and center of gravity position is presented. However, \cite{maki2020model} utilized the Euler angle to represent the attitude of the quadrotor, which will lead to loss of control due to singularity issues \cite{chaturvedi2011rigid} under large impact disturbances. In \cite{hanover2021performance}, $\mathcal{L}_1$ adaptive controller is employed to directly estimate and compensate for disturbances at the rotor thrust level, enhancing the disturbance rejection performance of the basic nonlinear model predictive control (NMPC) with minimal computational overhead.

In this work, to tackle the control challenges of the quadrotor under instantaneous impact and payload disturbances, we propose a Dual-Channel Adaptive Nonlinear Model Predictive Control (DCA-NMPC) framework that cascades NMPC with two lower-level model reference adaptive controllers, which combines the fast convergence and predictive capabilities of NMPC with the parameter adaptive abilities of MRAC, enabling the quadrotor to recover quickly. The primary contributions of this paper are outlined below:

\begin{itemize}
	\item  A nonsingular control approach based on quaternion attitude representation is proposed. By integrating quaternion attitude as part of the state into NMPC for optimization, it avoids the singularity issues encountered by controllers based on Euler angle representation \cite{maki2020model}.
	\item Explicit handling of uncertainties, rather than compensating for a lump disturbance term \cite{hassani2024model, xu2024novel}. Disturbances in translational dynamics are primarily caused by unknown payload mass, which are addressed by the proposed Thrust Model Reference Adaptive Controller (Thrust MRAC). Additionally, the Angular Rate Model Reference Adaptive Controller (Angular Rate MRAC) is introduced to compensate for changes in the inertia matrix and center of mass position.
	\item The actuator time-delay effect is considered. In NMPC, a first-order model is used to approximate the angular velocity loop with the delay effect, reducing the computational and parameter identification burden compared to directly controlling rotor thrust \cite{hanover2021performance}. Furthermore, in MRAC, a modified reference model is employed to enhance the system's robustness effectively.
\end{itemize}

In addition, several simulation experiments show that the proposed controller can quickly stabilize under disturbances and recover rapidly even during trajectory tracking.

The paper is organized as follows: Sec. \ref{sec:Quadrotor Dynamics} presents the quadrotor dynamics model and performs disturbance analysis. Sec. \ref{sec:controller} describes the proposed controller architecture in detail. In Sec. \ref{sec:simulation}, numerical simulation experiments verify the disturbance rejection effect of the proposed controller. Finally, Sec. \ref{sec:conclusion} summarizes the paper.

\section{Quadrotor Dynamics and Disturbance Analysis}\label{sec:Quadrotor Dynamics}

\subsection{Notation}

To model the effects of payload on the dynamics of the quadrotor, three right-handed coordinate systems are introduced: the world frame ${W}$: $\{\boldsymbol{x}_W,\boldsymbol{y}_W,\boldsymbol{z}_W\}$, the body frame $B$: $\{\boldsymbol{x}_B,\boldsymbol{y}_B,\boldsymbol{z}_B\}$ and the center of mass frame $M$ :$\{\boldsymbol{x}_M,\boldsymbol{y}_M,\boldsymbol{z}_M\}$, as shown in Fig. \ref{fig:frame}. The origin of the $B$ frame is defined at the geometric center of the body $x,y$-plane, with the $\boldsymbol{x}_B$ pointing forward and the $\boldsymbol{z}_B$ in the same direction as the collective thrust direction. The $M$ frame is located at the center of mass of the quadrotor, with the axes oriented in the same direction as the $B$ frame. For vector notation, the coordinate of point $Y$ in coordinate frame $X$ is represented by $\boldsymbol{r}_{XY}$. Unless otherwise specified, numerical subscripts are used to indicate the components in the corresponding positions, namely, $(\cdot)_{i}$ is the $i$-th component of the vector, and $(\cdot)_{ij}$ represents the element at the $i$-th row and $j$-th column of the matrix. Furthermore, the quadrotor orientation is represented by the unit quaternion $\boldsymbol{q}=[q_w,q_x,q_y,q_z]^{\intercal}\in\mathbb{S}^3$, and $\otimes$ represents the quaternion multiplication.
\subsection{Quadrotor Dynamics}
Assuming that the quadrotor is the rigid body, the quadrotor dynamics model can be obtained using the theory of rigid body dynamics \cite{sun2022comparative} as
\begin{align}
\dot{\boldsymbol{p}}&=\boldsymbol{v} \label{qmodel_1},
\\\dot{\boldsymbol{q}}&=\boldsymbol{q}\otimes\begin{bmatrix}0&\boldsymbol{\omega}^{\intercal}/2\end{bmatrix}^{\intercal} \label{qmodel_2},
\\\dot{\boldsymbol{v}}&={T\boldsymbol{z}_{B}}/m+\boldsymbol{g} \label{qmodel_3},
\\\dot{\boldsymbol{\omega}}&=
\boldsymbol{J}^{-1}[\boldsymbol{\tau}-\boldsymbol{\omega}\times \boldsymbol{J}\boldsymbol{\omega}] \label{qmodel_4},
\end{align}
where $\boldsymbol{p},\boldsymbol{v}$ denote the position and velocity of the geometric center of the quadrotor, $T,\boldsymbol{\tau}$ represent the collective thrust and body torque, $\boldsymbol{\omega}$ is the angular rate of $B$ frame with respect to $W$ frame, $\boldsymbol{J}$ denotes the inertia matrix of the quadrotor, $m$ is the mass of the quadrotor, and $\boldsymbol{g}=[0,0,-g]^\intercal$ represents the gravity vector.

According to the rotor positions and spin configurations of the quadrotor denoted in Fig \ref{fig:frame}, the relationship between the collective thrust $T$, the body torque $\boldsymbol{\tau}$ and the thrust of a single rotor can be expressed as
\begin{equation}
\left[\begin{array}{l}
	T \\
	\boldsymbol{\tau} \\
\end{array}\right]=\boldsymbol{G}\mathbf{f},
\end{equation}
where $\mathbf{f}=[f_1,f_2,f_3,f_4]^\intercal$,  $f_i$, $i=1,2,3,4$ represents the thrust generated by each rotor, and $\boldsymbol{G}$ is the control effectiveness matrix:
\begin{equation}
\boldsymbol{G}=\left[\begin{array}{cccc}
1 & 1 & 1 & 1 \\
-d_y& -d_y & d_y & d_y \\
-d_x & d_x & d_x & -d_x \\
-c_{\tau} & c_{\tau} & -c_{\tau} & c_{\tau}
\end{array}\right],
\end{equation}
where $d_x,d_y$ represent the distance from the rotor to the $\boldsymbol{x}_B,\boldsymbol{y}_B$ axis, respectively. $c_{\tau}$ represents the rotor drag torque constant.

Since the thrust generated by the rotor cannot track the input rotor thrust command instantaneously, similar to \cite{nan2022nonlinear}, the time-delay effect of the rotor actuator is modeled as a first-order model:
\begin{equation}
\dot{f}_i=\frac{1}{\sigma_a}(u_i-f_i),\quad i=1,2,3,4,
\end{equation}
where $\sigma_a$ is the rotor actuator time constant and $u_i$ is the thrust command input to the rotor.
\subsection{Effects of Impact Collision}
To analyze the effect of contact force generated by impact collision, assume that the quadrotor is a smooth cylinder and the payload is considered as a particle with a mass of $m_P$, which undergoes a completely inelastic collision with the quadrotor at the collision point $P$, as depicted in Fig. \ref{fig:frame}. Before the collision, the quadrotor center of mass frame $M$ and the body frame $B$ are aligned; therefore, during the collision analysis, $\boldsymbol{r}_{MP}$ is replaced by $\boldsymbol{r}_{BP}$ to avoid confusion.

The magnitude of the collision impulse is represented by ${I}$, and its direction points to the inside of the cylinder along the normal direction of the surface, represented by $\boldsymbol{n}_q$. According to the linear impulse-momentum theorem and the angular impulse-momentum theorem \cite{gilardi2002literature}, one can obtain that
\begin{align}
\boldsymbol{v^+}&=\boldsymbol{v^-}+\frac{I}{m}\boldsymbol{n}_q,\label{impact:v}\\
\boldsymbol{\omega^+}&=\boldsymbol{\omega^-}+I\boldsymbol{J}^{-1}\cdot(\boldsymbol{r}_{BP}\times\boldsymbol{n}_q),\label{impact:omega}
\end{align}
where the superscript $-$ indicates variables before the collision, and the superscript $+$ indicates variables after the collision. For the detailed analysis process and results, please refer to Appendix \ref{app:Impact}.

\subsection{Effects of Payload Uncertainty}
After the quadrotor collides with the payload, the payload is attached to the quadrotor body to form a quadrotor-payload system. Ignoring the relative acceleration of the $B$ frame and the $M$ frame, the system dynamics are similar to (\ref{qmodel_1}) to (\ref{qmodel_4}), except that (\ref{qmodel_3}) and (\ref{qmodel_4}) are modified as
\begin{align}
\dot{\boldsymbol{v}}&={T\boldsymbol{z}_{B}}/(m+m_P)+\boldsymbol{g} \label{smodel_3},
\\\dot{\boldsymbol{\omega}}&=
\boldsymbol{J}_S^{-1}[\boldsymbol{\tau}+\boldsymbol{\tau}_{CoM}-\boldsymbol{\omega}\times \boldsymbol{J}_S\boldsymbol{\omega}] \label{smodel_4},
\end{align}
where $\boldsymbol{J}_{S}$ is the inertia matrix of the quadrotor-payload system relative to the geometric center of the quadrotor, which is a symmetric semi-positive definite matrix, and $\boldsymbol{\tau}_{CoM}$ is the moment generated by gravity at the geometric center of the quadrotor, which is analyzed in the Appendix \ref{app:Payload}.

In summary, the impact collision induces abrupt changes in the velocity and angular rate of the quadrotor, while payload uncertainty causes disturbances within both translational and rotational dynamics.

\begin{figure}[bp]
	\vspace{-10pt}
	\centering
	\includegraphics[width=0.489\textwidth, trim=11cm 0cm 0.5cm 0cm, clip]{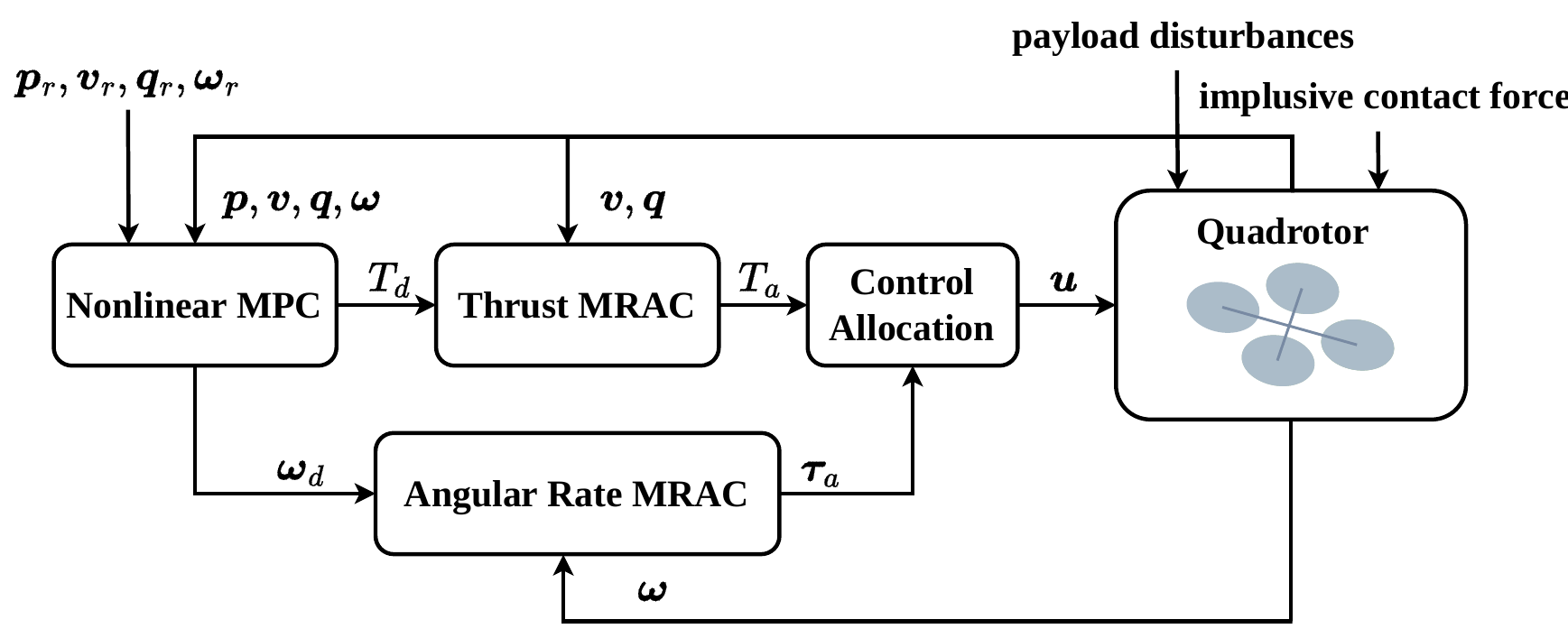}%
	\caption{Block diagram of proposed controller}
	\label{framework_fig}
\end{figure}

\section{Control Architecture}\label{sec:controller}
In this section, we propose a controller architecture called Dual-Channel Adaptive Nonlinear Model Predictive Control (DCA-NMPC) that cascades a nonlinear MPC to two adaptive controllers, which handle disturbances in translational dynamics and rotational dynamics, respectively. The block diagram of the proposed controller is shown in Fig. \ref{framework_fig}. The input-output relationship of each module is explained as follows:
\begin{itemize}
\item Nonlinear MPC: Outputing desired collective thrust command $T_d$ to Thrust MRAC and desired angular rate $\boldsymbol{\omega}_d$ to Angular Rate MRAC based on state information.
\item  Thrust MRAC: Scaling the desired collective thrust command $T_d$ to obtain the adaptive collective thrust command $T_a$.
\item Angular Rate MRAC: Outputing adaptive torque command $\boldsymbol{\tau}_a$ to track the desired angular rate $\boldsymbol{\omega}_d$.
\item Control Allocation: Allocating the adaptive combined thrust command $T_d$ and the adaptive torque command $\boldsymbol{\tau}_a$ to the rotor thrust $\boldsymbol{u}$.
\end{itemize}

\subsection{Nonlinear Model Predictive Control}

The state of the quadrotor is defined as $\boldsymbol{x}=[\boldsymbol{p}^\intercal,\boldsymbol{q}^\intercal,\boldsymbol{v}^\intercal,\boldsymbol{\omega}^\intercal]^\intercal$, and the input is $\boldsymbol{c}_d=[T_d,\boldsymbol{\omega_c}^\intercal]^\intercal$, where $T_d$ is the desired collective thrust command and $\boldsymbol{\omega_c}$ is the angular rate control command.

To take into account the tracking speed of the lower-level controller in the NMPC, the angular rate loop is approximated as the following first-order system:
\begin{equation}
\dot{\boldsymbol{\omega}}=-\frac{1}{\sigma_{rate}}(\boldsymbol{\omega}-\boldsymbol{\omega_c}) \label{ratemodel},
\end{equation}
where $\sigma_{rate}$ is the time constant of the angular rate loop, which is generally related to the reference system of the low-level adaptive controller.

The explicit 4th-order Runge-Kutta method is applied to discretize the system (\ref{qmodel_1}) to (\ref{qmodel_3}) and (\ref{ratemodel}) into $N$ time steps over the time horizon $H_t$, namely, taking the discretization step size $\delta_t=\frac{H_t}{N}$. The nonlinear optimal control problem is described as
\begin{equation}
\begin{aligned}
\min_{\boldsymbol{u}}\ \boldsymbol{\bar{x}}_{N}^{T}\boldsymbol{Q}\boldsymbol{\bar{x}}_{N}+\sum_{k=0}^{N-1}\boldsymbol{\bar{x}}_{k}^{T}\boldsymbol{Q}\boldsymbol{\bar{x}}_{k}+\bar{\boldsymbol{c}}_{d,k}^{T}\boldsymbol{R}\bar{\boldsymbol{c}}_{d,k}\\
\mathrm{subject~to}\quad\boldsymbol{x}_{k+1}=\boldsymbol{f}_{D}(\boldsymbol{x}_{k},\boldsymbol{c}_{d,k},\delta t)\\
\boldsymbol{x}_{0}=\boldsymbol{x}_{init} \quad
\boldsymbol{c}_{d,k}\in[\boldsymbol{c}_{min},\boldsymbol{c}_{max}]
\end{aligned}
\end{equation}
where $k$ represents the current time step, and $f_{D}$ denotes the discretized dynamic model. The notation $\overline{(\cdot)}:=(\cdot)-(\cdot)_{ref}$ represents the error between the current state and the reference state from the trajectory planner. $\boldsymbol{Q}:=\mathrm{diag}\left(\boldsymbol{Q}_p,\boldsymbol{Q}_v,\boldsymbol{Q}_q,\boldsymbol{Q}_\omega\right)$ is the positive-definite weight matrix of the state cost, and $\boldsymbol{R}$ represents the positive-definite weight matrix of the input cost. $\boldsymbol{c}_{min},\boldsymbol{c}_{max}$ are the lower and upper bounds on the input, respectively. Note that due to the particularity of the manifold $\mathbb{S}^{3}$ where the quaternion is located, the quaternion term in the cost function can be calculated by the following \cite{li2023nonlinear}
\begin{equation}
\bar{\boldsymbol{q}}_k^T\boldsymbol{Q}_q\bar{\boldsymbol{q}}_k=\mathrm{vec}(\boldsymbol{q}_e)^T\boldsymbol{Q}_q\mathrm{vec}(\boldsymbol{q}_e),
\end{equation}
where $\boldsymbol{q}_e=\boldsymbol{q}\otimes\boldsymbol{q}_{ref}^{-1}$ represents the error between quaternions, and $\mathrm{vec}(\cdot)$ means the vector part of the quaternion, namely, $\mathrm{vec}(\boldsymbol{q})=[q_x,q_y,q_z]^\intercal$. The above nonlinear optimal control problem can be solved in real time using the ACADOS toolkit \cite{verschueren2022acados}.

Although the solved $\boldsymbol{\omega_c}$ takes into account the delay of the angular rate loop, the change of $\boldsymbol{\omega_c}$ is too drastic to be tracked by the lower-level controller. Therefore, the angular rate $\boldsymbol{\omega}_{k+1}$ in the state $\boldsymbol{x_{k+1}}$ of the next time step is used as the desired angular rate signal $\boldsymbol{\omega_d}$ of the lower-level controller.

\begin{remark}
The nominal rotational dynamics model (\ref{qmodel_4}) is not utilized in NMPC because it is no longer valid under disturbances. Instead, the linear first-order model, being simple and computationally efficient, is more suitable for anti-disturbance scenarios. More importantly, it leverages the predictive capability of NMPC to mitigate the time-delay effect in the angular velocity loop with minimal overhead. Additionally, it eliminates the need to identify the parameters of the inertia matrix. 
\end{remark}

\subsection{Thrust Model Reference Adaptive Control}
As shown before, the thrust model reference adaptive control is used to reject disturbances in the collective thrust direction mainly caused by the unknown payload mass, as demonstrated by (\ref{smodel_3}). Scaling the desired thrust command from NMPC is considered for disturbance compensation:
\begin{equation}
T_a=k_tT_d \label{scaling},
\end{equation}
where $T_a$ is the adaptive collective thrust command, and $k_t$ is the scaling factor introduced to counteract the uncertainty of payload mass.

The nominal model (\ref{qmodel_3}) is adopted as the reference model:
\begin{equation}
\dot{\boldsymbol{v}}_m={T_d\boldsymbol{z}_{B}}/m+\boldsymbol{g}+\underbrace{\boldsymbol{L}_t\cdot({\boldsymbol{v}}-{\boldsymbol{v}}_m)}_{\text{Error Feedback Term}} \label{v_ref},
\end{equation}
In addition, the observer-like error feedback term $\boldsymbol{L}_t\cdot({\boldsymbol{v}}-{\boldsymbol{v}}_m)$ is added, where $\boldsymbol{L}_t$ is a positive definite diagonal matrix, which can improve transient performance \cite{lavretsky2012robust}.

The adaptive law is constructed as
\begin{equation}
\dot{k}_t=-\Gamma_t\boldsymbol{e}_v^\intercal\boldsymbol{t}_d
\label{kt_ada},
\end{equation}
where $\Gamma_t$ is the adaptation rate of $k_t$, ${\boldsymbol{e}}_v:=\boldsymbol{v}-\boldsymbol{v}_m$ represents the error between the reference model and the actual system model, and $\boldsymbol{t}_d=T_d\boldsymbol{z}_{B}$ is the desired thrust vector.

\begin{theorem}
The tracking error $\boldsymbol{e}_v$ between the reference system model (\ref{v_ref}) and the actual system model (\ref{smodel_3}) with the collective thrust control law (\ref{scaling}) and the adaption law (\ref{kt_ada}) asymptotically converges to zero.
\label{thm:1}
\end{theorem}
\begin{proof}
See Appendix \ref{app:1}.
\end{proof}

\begin{remark}
Different from the general designs of MRAC, the nominal translational dynamics (\ref{qmodel_3}) is utilized as the reference system. The purpose here is to ensure that the actual system's response matches the response of the model within the superior-level NMPC controller, thereby indirectly compensating for the effects of external disturbances. Additionally, because the magnitude is scaled, it is actually capable of resisting unknown disturbances in the direction of the collective thrust.
\end{remark}

\subsection{Angular Rate Model Reference Adaptive Control}
After loading, the uncertainty of the inertia matrix $\boldsymbol{J}_{S}$ and gravity moment $\boldsymbol{\tau}_{CoM}$ of the quadrotor-payload system will cause the baseline control performance to degrade rapidly. Inspired by \cite{maki2020model}, a Multiple-Input Multiple-Output angular rate model reference adaptive controller is proposed.

A first-order reference system with an observer-like error feedback term is introduced as
\begin{equation}
\dot{\boldsymbol{\omega}}_m=\boldsymbol{K}_\omega({\boldsymbol\omega}_d-\boldsymbol\omega_m)+\boldsymbol{L}_\omega({\boldsymbol\omega}-\boldsymbol\omega_m)
\label{omega_ref},
\end{equation}
where $\boldsymbol{K}_\omega$ and $\boldsymbol{L}_\omega$ are both positive definite diagonal weight matrices. $\boldsymbol{K}_\omega$ denotes the tracking speed of the reference signal to the desired one, while $\boldsymbol{L}_\omega$ takes into account the tracking capability of the actual system, avoiding the instantaneous increase of the tracking error and causing control divergence.

The unknown parameter vector to contain elements of the system's inertia matrix and gravity moment is defined as
\begin{equation}
\boldsymbol{\gamma}^*=[J_{S,11},J_{S,22},J_{S,33},J_{S,12},J_{S,23},J_{S,13},\tau_{CoM,1},\tau_{CoM,2}]^\intercal,
\end{equation}

Similar to \cite{maki2020model}, the cross product term $\boldsymbol{\omega}\times \boldsymbol{J}_{S}\boldsymbol{\omega}$ in (\ref{smodel_3}) is ignored and multiplied by $\boldsymbol{J}_{S}$ from the left to obtain the following simplified model:
\begin{equation}
\boldsymbol{J}_{S}\dot{\boldsymbol{\omega}}=\boldsymbol{\tau}_a+\boldsymbol{\tau}_{CoM}
\label{simplified_rot_model},
\end{equation}
where $\boldsymbol{\tau}_a$ is the adaptive body torque command.

Define ${\boldsymbol{e}_\omega}:=\boldsymbol{\omega}-\boldsymbol{\omega}_m$ and introduce the following representation matrix:
\begin{equation}
\boldsymbol{\varphi}_\omega^T=\begin{bmatrix}\phi_1&0&0&\phi_2&0&\phi_3&1&0\\0&\phi_2&0&\phi_1&\phi_3&0&0&1\\0&0&\phi_3&0&\phi_2&\phi_1&0&0\end{bmatrix}
\label{represent},
\end{equation}
where $\boldsymbol{\phi}=\boldsymbol{K}_\omega({\boldsymbol\omega}-\boldsymbol\omega_d)$. 

The following adaptive law is introduced as
\begin{equation}
\dot{\boldsymbol{\gamma}}=\boldsymbol{\Gamma}_\omega\boldsymbol{\varphi}_\omega\boldsymbol{e}_\omega 
\label{old_law_rate},
\end{equation}
where $\boldsymbol{\gamma}$ is the estimation of the unknown parameter vector $\boldsymbol{\gamma}^*$, and $\boldsymbol{\Gamma}_\omega$ is the adaptation rate of $\boldsymbol{\gamma}$.

Design the following control law:
\begin{equation}
\boldsymbol{\tau}_a=-\boldsymbol{\varphi}_\omega^T\boldsymbol{\gamma}-\boldsymbol{K}_p\boldsymbol{e}_\omega
\label{rot_control},
\end{equation}
where a proportional control term is also added to accelerate error convergence. $\boldsymbol{K}_p$ is a positive definite diagonal matrix, representing the proportional control gain.

\begin{theorem}
The tracking error $\boldsymbol{e}_\omega$ between the reference system model (\ref{omega_ref}) and the simplified system model (\ref{simplified_rot_model}) with the adaptive body torque control law (\ref{rot_control}) and the adaption law (\ref{old_law_rate}) asymptotically converges to zero.
\label{thm:2}
\end{theorem}
\begin{proof}
	See Appendix \ref{app:2}.
\end{proof}

In order to avoid adaptive parameter drift, the e-modification technique \cite{narendra1987new} is applied, and (\ref{old_law_rate}) is improved as
\begin{equation}
\dot{\boldsymbol{\gamma}}=\boldsymbol{\Gamma}_\omega\boldsymbol{\varphi}_\omega\boldsymbol{e}_\omega-\boldsymbol{\mu}\Vert \boldsymbol{e}_\omega \Vert\boldsymbol{\gamma},
\end{equation}
where $\boldsymbol{\mu}>0$ represents the e-modification parameter, and $\Vert \cdot \Vert$ represents the Euclidean norm of the vector. e-modification adds a damping term to the adaptive law, which can effectively enhance the robustness of the adaptive control.

\begin{remark}
Compared to the attitude adaptive controller presented in \cite{maki2020model}, the proposed control architecture delegates attitude control to the superior-level controller, which employs quaternion representation to avoid the singularity issue associated with Euler angle representation in \cite{maki2020model}. 
\end{remark}

\begin{remark}
	 As shown in \cite{xu2024novel}, considering actuator dynamics is crucial for control performance, and actuator time-delay effect also limits the adaptation rate \cite{nguyen2018model}. The error feedback term is added to the reference system, balancing the system's response speed and robustness. It allows the controller to achieve reliable control with a small adaptation rate, avoiding the degradation of control performance caused by the actuator time-delay effect under a high adaptation rate.
\end{remark}

\subsection{Control Allocation}
The control allocation module is responsible for mapping the adaptive collective thrust command and adaptive body torque command to the thrust of each rotor. Here the control allocation module based on quadratic programming (QP) is used, similar to \cite{sun2022comparative}:
\begin{equation}
\begin{aligned}
&\min_{\boldsymbol{c}}\ \boldsymbol{e}_f^\intercal \boldsymbol{Q}_f\boldsymbol{e}_f+\boldsymbol{u}^{T}\boldsymbol{R}_f\boldsymbol{u}\\
\mathrm{subject~to}&\quad{u}_{min}\leq u_{i}\leq u_{max},i=1,2,3,4
\end{aligned}
\end{equation}
where $\boldsymbol{e}_f=\boldsymbol{G}\mathbf{u}-[T_a,\boldsymbol{\tau}_a^\intercal]^\intercal$ is the control allocation error. $u_{min}$ and $u_{max}$ are the minimum and maximum values of the rotor thrust input, respectively. The calculated thrust will be sent to the quadrotor for execution.

\section{Numerical Simulation}\label{sec:simulation}
Two types of simulation testing scenarios are provided: static anti-disturbance scenarios and dynamic anti-disturbance scenarios. The static anti-disturbance scenarios test the ability of the quadrotor to quickly stabilize after being disturbed in the hovering state, while the dynamic scenarios test the ability of the quadrotor to quickly recover while tracking the trajectory.
\subsection{Simulation Setup}
The explicit 4th-order Runge-Kutta algorithm is used to simulate the quadrotor dynamics with a step size of 0.5 ms, and the impact collision effects of the payload, the payload uncertainty, and the rotor actuator time-delay effect are simulated according to Sec. \ref{sec:Quadrotor Dynamics}, with $\sigma_a$ set to 25 ms.
	
The nonlinear model reference adaptive control (abbreviated as NMRAC) in \cite{maki2020model} and the $\mathcal{L}_1$ adaptive controller (abbreviated as $\mathcal{L}_1$) in \cite{hanover2021performance} are implemented for comparison. In addition, the control parameters of the proposed controller's NMPC are kept consistent with \cite{hanover2021performance} for fairness.

The total simulation duration is uniformly set to 10 seconds. The payload is set to make a free fall motion with an initial velocity of zero from a height of $h_P$ m. At the 3rd second from the start of the simulation, the payload collides with the quadrotor at the position $(r_{BP,1},r_{BP,2})$ m of the quadrotor's $x,y$-plane. After the collision, the payload is attached to the quadrotor's $x,y$-plane and no longer moves relative to the quadrotor. Considering the length of the quadrotor's arms, $(r_{BP,1},r_{BP,2})$ is uniformly set to $(0.2,0.2)$ m, and let different mass payloads fall from different heights to hit the quadrotor to compare the anti-disturbance performance of controllers.

Furthermore, to simulate real sensor data, Gaussian noise is added to the quadrotor state obtained by the controller. Note that the exponential mapping \cite{sola2017quaternion} is employed to simulate noise in attitude.

\subsection{Simulation Experiments}
\subsubsection{Static Anti-disturbance Scenarios}
To evaluate the proposed controller's ability to quickly stabilize near the equilibrium point, the quadrotor is initially set to a hovering state, and after being impacted, it needs to control its velocity and attitude to resume hovering quickly. Note that it is not mandatory for the quadrotor to return to its initial position here.

The mean absolute error of velocity $\mathrm{MAE}_v$ is used to quantitatively measure the control performance as
	\begin{equation}
	\mathrm{MAE}_v:=\frac{\sum_{k=1}^{N_{sim}}\Vert\boldsymbol{v}_{k}-\boldsymbol{v}_{ref,k}\Vert}{N_{sim}},
	\end{equation}
where $N_{sim}$ is the total number of time steps of the simulation. The smaller $\mathrm{MAE}_v$, the smaller the average error of velocity in time.

\begin{table}[tbp]
	\centering
	\caption{Comparison of velocity mean absolute error in static anti-disturbance scenarios}
	\label{table:static}
	\begin{tabular}{|cc|ccc|}
		\hline
		\multicolumn{2}{|c|}{Setup}                            & \multicolumn{3}{c|}{$\mathrm{MAE}_v$ [cm/s]}        \\ \hline
		\multicolumn{1}{|c|}{$m_P$[kg]}             & $h_P$[m] & \multicolumn{1}{c|}{this paper}            & \multicolumn{1}{c|}{$\mathcal{L}_1$ \cite{hanover2021performance}} & NMRAC \cite{maki2020model}  \\ \hline
		\multicolumn{1}{|c|}{\multirow{3}{*}{0.2}} & 0.2      & \multicolumn{1}{c|}{5.776}           & \multicolumn{1}{c|}{\textbf{5.132}}  & 12.418 \\ \cline{2-5} 
		\multicolumn{1}{|c|}{}                      & 0.5      & \multicolumn{1}{c|}{7.237}           & \multicolumn{1}{c|}{\textbf{6.367}}  & 14.736 \\ \cline{2-5} 
		\multicolumn{1}{|c|}{}                      & 0.8      & \multicolumn{1}{c|}{\textbf{8.535}}  & \multicolumn{1}{c|}{8.951}           & 18.137 \\ \hline
		\multicolumn{1}{|c|}{\multirow{3}{*}{0.5}}  & 0.2      & \multicolumn{1}{c|}{\textbf{9.865}}  & \multicolumn{1}{c|}{12.837}          & 16.706 \\ \cline{2-5} 
		\multicolumn{1}{|c|}{}                      & 0.5      & \multicolumn{1}{c|}{\textbf{14.398}} & \multicolumn{1}{c|}{21.613}          & 36.389 \\ \cline{2-5} 
		\multicolumn{1}{|c|}{}                      & 0.8      & \multicolumn{1}{c|}{\textbf{16.584}} & \multicolumn{1}{c|}{35.285}          & 55.338 \\ \hline
	\end{tabular}
\end{table}

\begin{figure}[tbp]
	\centering
	\includegraphics[scale=0.5, trim=0.6cm 1.2cm 1.5cm 2.5cm, clip]{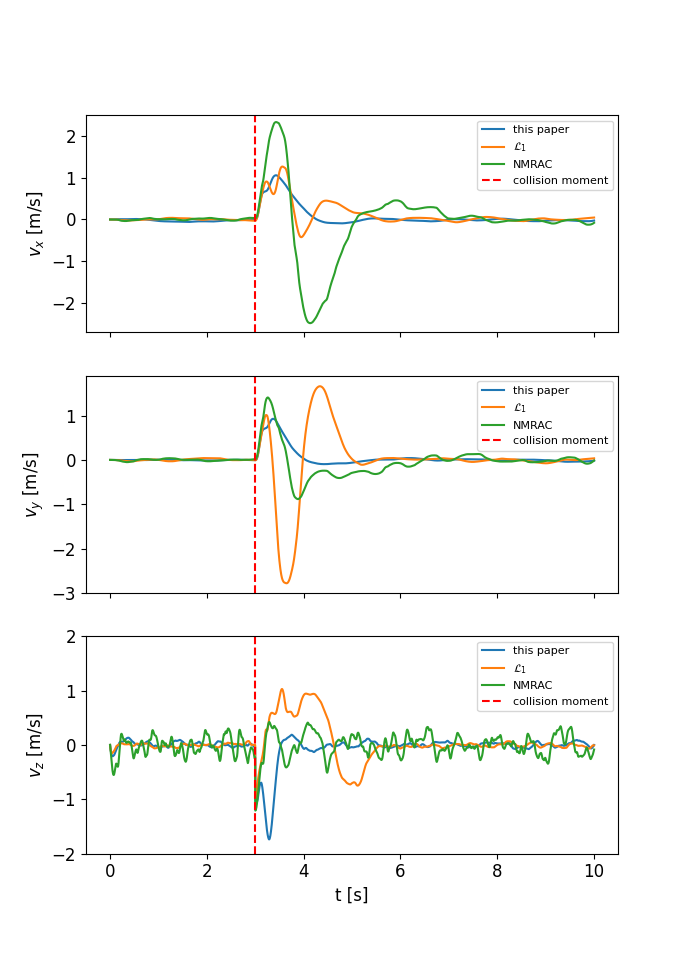}%
	\caption{Comparison results of velocity control performance of the controllers with a 0.5 kg payload dropped from 0.8 m height in static anti-disturbance scenarios}
	\label{fig:v_traj_static}
	\vspace{-10pt}
\end{figure}

Under different disturbance settings, the comparison results of the velocity control performance of various controllers are presented in Table \ref{table:static}. The proposed controller exhibits performance improvements compared to the other controllers in most settings. Only for the payload mass of 0.2 kg and the drop height of 0.2 m and 0.5 m, the mean absolute error of the proposed controller is slightly higher than that of the $\mathcal{L}_1$ controller. This may be due to the $\mathcal{L}_1$ controller directly outputting rotor thrust, resulting in finer control under small disturbances. Nevertheless, in other scenarios with larger disturbances, the proposed controller demonstrates significant performance enhancements. Specifically, when the payload mass is 0.5 kg and the drop height is 0.8 m, the $\mathrm{MAE}_v$ of the proposed controller is reduced by 53.00\% compared to the $\mathcal{L}_1$ controller and by 70.03\% compared to the NMRAC controller. Fig. \ref{fig:v_traj_static} displays the velocity control performance of various controllers in that case. The proposed controller converges in a smoother and faster manner compared to the other controllers. 
\begin{table}[tbp]
	\vspace{-15pt}
	\centering
	\caption{Comparison of position mean absolute error in dynamic anti-disturbance scenarios}
	\label{table:dynamic}
	\begin{tabular}{|cc|ccc|}
		\hline
		\multicolumn{2}{|c|}{Setup}                           & \multicolumn{3}{c|}{$\mathrm{MAE}_p$ [cm]}         \\ \hline
		\multicolumn{1}{|c|}{$m_P$[kg]}            & $h_P$[m] & \multicolumn{1}{c|}{this paper}           & \multicolumn{1}{c|}{$\mathcal{L}_1$ \cite{hanover2021performance}} & NMRAC \cite{maki2020model}  \\ \hline
		\multicolumn{1}{|c|}{\multirow{3}{*}{0.2}} & 0.2      & \multicolumn{1}{c|}{\textbf{2.687}} & \multicolumn{1}{c|}{3.317}           & 19.244 \\ \cline{2-5} 
		\multicolumn{1}{|c|}{}                     & 0.5      & \multicolumn{1}{c|}{\textbf{3.563}} & \multicolumn{1}{c|}{3.603}           & 19.101 \\ \cline{2-5} 
		\multicolumn{1}{|c|}{}                     & 0.8      & \multicolumn{1}{c|}{\textbf{4.076}} & \multicolumn{1}{c|}{4.135}           & 19.251 \\ \hline
		\multicolumn{1}{|c|}{\multirow{3}{*}{0.5}} & 0.2      & \multicolumn{1}{c|}{\textbf{3.654}} & \multicolumn{1}{c|}{6.424}           & failed \\ \cline{2-5} 
		\multicolumn{1}{|c|}{}                     & 0.5      & \multicolumn{1}{c|}{\textbf{5.701}} & \multicolumn{1}{c|}{11.931}          & failed \\ \cline{2-5} 
		\multicolumn{1}{|c|}{}                     & 0.8      & \multicolumn{1}{c|}{\textbf{8.994}} & \multicolumn{1}{c|}{14.680}          & failed \\ \hline
	\end{tabular}
	\vspace{-10pt}
\end{table}

\begin{figure}[tbp]
	\centering
	\begin{tabular}{@{\extracolsep{\fill}}c@{}c@{\extracolsep{\fill}}}
		\includegraphics[width=0.25\textwidth, trim=3.5cm 2.0cm 1.8cm 2.5cm, clip]{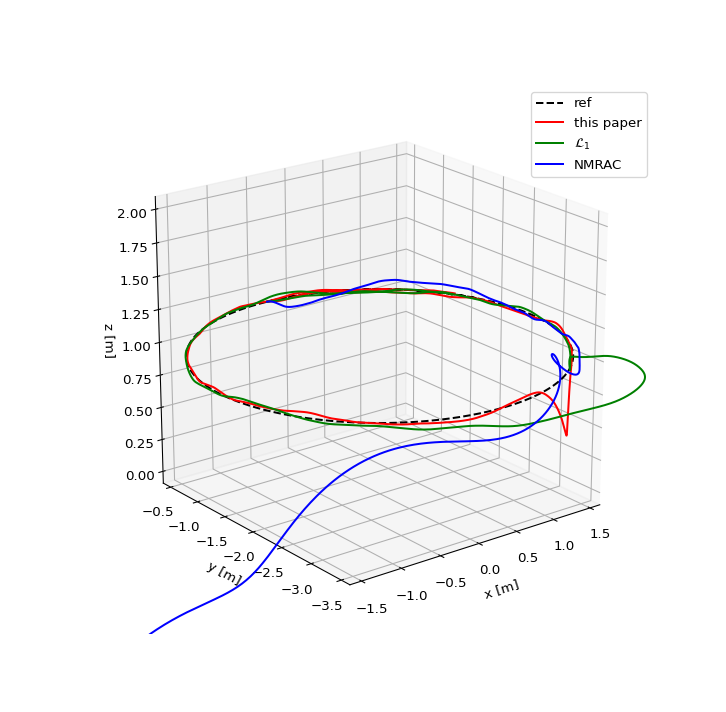} &
		\includegraphics[width=0.23\textwidth, trim=0.0cm 0.0cm 1.0cm 1cm, clip]{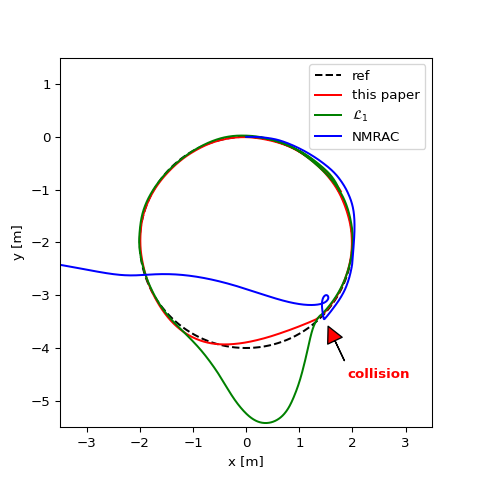}\\
		(a) 3D view & (b) upper view\\
	\end{tabular}
	\caption{Comparison results of tracking performance of the controllers with a 0.5 kg payload dropped from 0.8 m height in dynamic anti-disturbance scenarios}
	\label{fig:track_results}
	\vspace{-15pt}
\end{figure}
\subsubsection{Dynamic Anti-disturbance Scenarios}
To measure the anti-disturbance capability of the proposed controller when the quadrotor is tracking the trajectory, the quadrotor is set to track a circular trajectory with a radius of 2 m at a speed of 1.5 m/s. After colliding with the payload, the quadrotor needs to resist the disturbance and continue to track the circular trajectory. The mean absolute error $\mathrm{MAE}_p$ of the position over the simulation time is used as
	\begin{equation}
	\mathrm{MAE}_p:=\frac{\sum_{k=1}^{N_{sim}}\Vert\boldsymbol{p}_{k}-\boldsymbol{p}_{ref,k}\Vert}{N_{sim}},
	\end{equation}
	
Table \ref{table:dynamic} provides the comparative results of the tracking performance of various controllers in dynamic anti-disturbance scenarios. The proposed controller demonstrates superior disturbance rejection performance even during motion compared to other controllers. Additionally, since NMRAC is based on Euler angles, under large payload disturbances, this leads to control divergence due to singularity issues. Fig. \ref{fig:track_results} illustrates the trajectory tracking results of various controllers when the payload mass is 0.5 kg and the drop height is 0.8 m. For the $\mathcal{L}_1$ controller, the error at the moment of collision does not converge quickly, resulting in a significant deviation from the reference trajectory. In contrast, the proposed controller responds rapidly, compensating for the disturbance before the error continues to increase, thereby enabling a faster recovery.

\section{Conclusion}\label{sec:conclusion}
This paper proposes to use NMPC as a superior-level controller to overcome impact disturbances. Thrust MRAC and Angular Rate MRAC serve as lower-level controllers, designed to resist disturbances in the direction of the collective thrust and disturbances caused by parametric uncertainties in rotational dynamics, respectively. The superior disturbance rejection performance of the proposed controllers has been verified in different simulation scenarios.

\appendix
\subsection{Disturbance Analysis}\label{app:disturbance}
\subsubsection{Analysis of Impact Collision}\label{app:Impact}
According to the linear impulse-momentum theorem and the angular impulse-momentum theorem \cite{gilardi2002literature}, one can obtain that
\begin{align}
m\left(\boldsymbol{v^+}-\boldsymbol{v^-}\right)&={I}\boldsymbol{n}_q \label{collision_1},\\
m_{P}\left(\boldsymbol{v}^+_{P}-\boldsymbol{v}^-_{P}\right)&=-{I}\boldsymbol{n}_q \label{collision_2},\\
\boldsymbol{J}\cdot(\boldsymbol{\omega^+}-\boldsymbol{\omega^-})&={I}\boldsymbol{r}_{BP} \times \boldsymbol{n}_q  \label{collision_3},
\end{align}
where $\boldsymbol{v}_P$ represents the velocity of the payload.

Assume that a perfectly inelastic collision occurs at the collision point, that is, the velocity of the quadrotor at the collision point is the same as the payload velocity after the collision:
\begin{equation}
\boldsymbol{v^+}+\boldsymbol{\omega^+}\times \boldsymbol{r}_{BP}=\boldsymbol{v}^+_{P} \label{collision_4}.
\end{equation}
By combining (\ref{collision_1}) to (\ref{collision_4}), it yields
\begin{equation}
I=\frac{1}{\mu}(\boldsymbol{v}^-_{P}-\boldsymbol{v^-}-\boldsymbol{\omega^-}\times \boldsymbol{r}_{BP})\cdot\boldsymbol{n}_q,
\end{equation}
where $\mu=\frac1m+\frac1{m_P}+[\boldsymbol{J}^{-1}\cdot(\boldsymbol{r}_{BP}\times\boldsymbol{n}_q)\times\boldsymbol{r}_{BP}]\cdot\boldsymbol{n}_q$.

The velocity and angular rate of the quadrotor after the collision can be obtained as (\ref{impact:v}) and (\ref{impact:omega}), respectively.

\subsubsection{Analysis of Payload Uncertainty}\label{app:Payload}
Due to the shift of the center of mass, gravity will generate a torque $\boldsymbol{\tau}_{CoM}$ on the geometric center of the quadrotor as
\begin{equation}
\boldsymbol{\tau}_{CoM}=\boldsymbol{r}_{BM}\times(m+m_P)\boldsymbol{g}\\=\begin{bmatrix}-r_{BM,2}\\
r_{BM,1}\\
0\end{bmatrix}\cdot(m+m_P)g,
\end{equation}
where $\boldsymbol{r}_{BM}$ can be obtained according to the definition of the center of mass:
\begin{equation}
\boldsymbol{r}_{BM}=\frac{m}{m+m_P}\boldsymbol{r}_{BP}.
\end{equation}

\subsection{Proof of Theorem \ref{thm:1}}\label{app:1}
\begin{proof}
Substituting (\ref{scaling}) into (\ref{smodel_3}) and subtracting (\ref{v_ref}), one can obtain that
\begin{equation}
\dot{\boldsymbol{e}}_v=-\boldsymbol{L}_t\cdot\boldsymbol{e}_v+b\tilde{k}_t\boldsymbol{t}_d
\label{e_v_dot}
\end{equation}
where $\tilde{k}_t={k}_t-{k}_t^*$, $k_t^*=\frac{m+m_P}{m}$ means the ideal gain when the tracking error converges and  $b=\frac{1}{m+m_P}>0$ is the unknown parameter.

Let us consider the following Lyapunov function as
\begin{equation}
V_t=\frac12\boldsymbol{e}_v^\intercal\boldsymbol{e}_v+\frac{b}{2\Gamma_t}\tilde{k}_t^2.
\end{equation}

Taking the derivative and substituting (\ref{kt_ada}) and (\ref{e_v_dot}) into it, it yields
\begin{equation}
\begin{aligned}
\dot{V}_t&=\boldsymbol{e}_v^\intercal(-\boldsymbol{L}_t\cdot\boldsymbol{e}_v+b\tilde{k}_t\boldsymbol{t}_d)+\frac{b}{\Gamma_t}\tilde{k}_t(-\Gamma_t\boldsymbol{e}_v^\intercal\boldsymbol{t}_d)\\&=-\boldsymbol{e}_v^\intercal\cdot\boldsymbol{L}_t\cdot\boldsymbol{e}_v<0.
\end{aligned}
\end{equation}

The asymptotic convergence of the tracking error can then be proved by applying Barbalat's lemma \cite{nguyen2018model}.
\end{proof}

\subsection{Proof of Theorem \ref{thm:2}}\label{app:2}
\begin{proof} 
Multiplying both sides of \ref{omega_ref} by $\boldsymbol{J}_{S}$ on the left gives
\begin{equation}
\boldsymbol{J}_{S}\dot{\boldsymbol{\omega}}_m=\boldsymbol{J}_{S}\boldsymbol{K}_\omega({\boldsymbol\omega}_d-\boldsymbol\omega_m)+\boldsymbol{J}_{S}\boldsymbol{L}_\omega{\boldsymbol{e}}_\omega.
\label{JS_ref}
\end{equation}
Subtracting (\ref{JS_ref}) from (\ref{simplified_rot_model}), it can be derived that
\begin{equation}
\boldsymbol{J}_{S}\dot{{\boldsymbol{e}}}_\omega=-\boldsymbol{J}_{S}\boldsymbol{L}_\omega{{\boldsymbol{e}}}_\omega-\boldsymbol{J}_{S}\boldsymbol{K}_\omega({\boldsymbol\omega}_d-\boldsymbol\omega_m)+\boldsymbol{\tau}_a+\boldsymbol{\tau}_{CoM}.
\end{equation}
Adding $\boldsymbol{J}_{S}\boldsymbol{K}_\omega{{\boldsymbol{e}}}_\omega$ on both sides gives
\begin{equation}
\boldsymbol{J}_{S}\dot{{\boldsymbol{e}}}_\omega+\boldsymbol{J}_{S}\boldsymbol{K}_\omega{{\boldsymbol{e}}}_\omega=-\boldsymbol{J}_{S}\boldsymbol{L}_\omega{{\boldsymbol{e}}}_\omega+\boldsymbol{J}_{S}\boldsymbol{K}_\omega({\boldsymbol\omega}-\boldsymbol\omega_d)+\boldsymbol{\tau}_a+\boldsymbol{\tau}_{CoM}. \label{mrac_rate_model}
\end{equation}

From (\ref{represent}), one can obtain that
\begin{equation}
\boldsymbol{\varphi}_\omega^T\boldsymbol{\gamma}^*=\boldsymbol{J}_{S}\boldsymbol{\phi}+\boldsymbol{\tau}_{CoM}.
\label{rate_mrac_dot_e_omega}
\end{equation}
Substituting (\ref{rate_mrac_dot_e_omega}) into (\ref{mrac_rate_model}), it can be obtained that
\begin{equation}
\boldsymbol{J}_{S}\dot{{\boldsymbol{e}}}_\omega+\boldsymbol{J}_{S}\boldsymbol{K}_\omega{{\boldsymbol{e}}}_\omega=-\boldsymbol{J}_{S}\boldsymbol{L}_\omega{{\boldsymbol{e}}}_\omega+\boldsymbol{\tau}_a+\boldsymbol{\varphi}_\eta^T\boldsymbol{\gamma}^*. \label{rate_mrac_2}
\end{equation}
Define $\tilde{\gamma}=\gamma-\gamma^*$ and introduce the following Lyapunov function:
\begin{equation}
V_\omega=\frac12\boldsymbol{e}_\omega^\intercal\boldsymbol{J}_{S}\boldsymbol{e}_\omega+\frac12\tilde{\gamma}^\intercal\Gamma_\omega^{-1}\tilde{\gamma}.
\end{equation}
Taking the derivative and substituting (\ref{rate_mrac_2}) into it, one has
\begin{equation}
\dot{V}_\omega=\boldsymbol{e}_\omega^\intercal(-\boldsymbol{J}_{S}\boldsymbol{K}_\omega{{\boldsymbol{e}}}_\omega-\boldsymbol{J}_{S}\boldsymbol{L}_\omega{{\boldsymbol{e}}}_\omega+\boldsymbol{\tau}_a+\boldsymbol{\varphi}_\omega^T\boldsymbol{\gamma}^*)+\tilde{\gamma}^\intercal\Gamma_\omega^{-1}\dot{{\gamma}}. \label{v_dot}
\end{equation}
Substituting (\ref{old_law_rate}) into (\ref{v_dot}), it can be deduced that
\begin{equation}
\begin{aligned}
\dot{V}_\omega&=-\boldsymbol{e}_\omega^\intercal\boldsymbol{J}_{S}\boldsymbol{K}_\omega{{\boldsymbol{e}}}_\omega-\boldsymbol{e}_\omega^\intercal\boldsymbol{J}_{S}\boldsymbol{L}_\omega{{\boldsymbol{e}}}_\omega-\boldsymbol{e}_\omega^\intercal\boldsymbol{K}_p\boldsymbol{e}_\omega-\boldsymbol{e}_\omega^\intercal\boldsymbol{\varphi}_\omega^T\tilde{\boldsymbol{\gamma}}\\&+\tilde{\gamma}^\intercal\boldsymbol{\varphi}_\omega\boldsymbol{e}_\omega\\
&=-\boldsymbol{e}_\omega^\intercal\boldsymbol{J}_{S}\boldsymbol{K}_\omega{{\boldsymbol{e}}}_\omega-\boldsymbol{e}_\omega^\intercal\boldsymbol{J}_{S}\boldsymbol{L}_\omega{{\boldsymbol{e}}}_\omega-\boldsymbol{e}_\omega^\intercal\boldsymbol{K}_p\boldsymbol{e}_\omega<0.
\end{aligned}
\end{equation}
Note that the positive definiteness of $\boldsymbol{J}_{S}$ is used here. Using Barbalat’s lemma, it can be shown that the tracking error converges asymptotically, but the convergence of the parameter vector $\gamma$ is not guaranteed here \cite{nguyen2018model}.
\end{proof}

\bibliographystyle{IEEEtran}
\bibliography{reference}

\end{document}